\documentclass[10pt]{llncs}

\usepackage{amssymb}
\usepackage{amsmath}
\usepackage{amsfonts}
\usepackage{mdframed}
\usepackage{graphicx}
\usepackage{subfigure}
\usepackage{dsfont}
\usepackage{caption}
\usepackage{geometry}
\usepackage{hyperref}
\usepackage{array}
\usepackage[hyphenbreaks]{breakurl}
\usepackage{geometry}
\begin{document}

\title{A Searchable Symmetric Encryption Scheme using BlockChain}

\author{Huige Li\inst{1 \and 2}\and Fangguo Zhang \inst{1 \and 2,}\thanks{Corresponding author, \email{isszhfg@mail.sysu.edu.cn}}\and Jiejie He\inst{1\and 2}\and Haibo Tian\inst{1,2,}}
\authorrunning{4}
\institute{School of Data and Computer Science, Sun Yat-Sen University, Guangzhou 510006, China\and Guangdong Key Laboratory of Information Security, Guangzhou 510006, China}
\maketitle

\begin{abstract}
At present, the cloud storage used in searchable symmetric encryption schemes (SSE) is provided in a private way, which cannot be seen as a true cloud. Moreover, the cloud server is thought to be credible, because it always returns the search result to the user, even they are not correct. In order to really resist this malicious adversary and accelerate the usage of the data, it is necessary to store the data on a public chain, which can be seen as a decentralized system. As
the increasing amount of the data, the search problem becomes more and more intractable, because there does not exist any effective solution at present.

In this paper, we begin by pointing out the importance of storing the data in a public chain. We then innovatively construct a model of SSE using blockchain(SSE-using-BC) and give its security definition to ensure the privacy of the data and improve the search efficiency. According to the size of data, we consider two different cases and propose two corresponding schemes. Lastly, the security and performance analyses show that our scheme is feasible and secure.
\end{abstract}

\keywords{
Searchable Encryption, Transaction, BlockChain, Cloud-Storage, Symmetric Encryption, Privacy.
}

\section{Introduction}

The cloud storage can allow users with limited physical resources to get access to their own data at any devices, and it only charges a few fee, therefore, more and more people prefer to upload their data onto the cloud. However, if the data are not processed before storage, the confidentiality and privacy of data cannot be guaranteed effectively. One feasible solution is to encrypt them by using a regular encryption algorithm before uploading. When the users want to retrieve the segments of the data, they download all the data and filter out that they need. However, if the number of
data that contain keyword $w$ is large, this method is unpractical. To address this issue, Song et al. firstly proposed Searchable Symmetric Encryption (SSE) term \cite{song2000practical}.

SSE aims to solve the search problem on ciphertexts. In this model, it usually involves three parties: Data owner, server and user. The data owner
 encrypts his $n$ documents $\mathbb{D}=\{D_{1},D_{2},\ldots, D_{ n}\}$ into ciphertexts $\mathbb{C}=\{C_{1},C_{2},\ldots,C_{ n}\}$. In order to improve the search efficiency, an auxiliary information, called as Index $\mathcal I$, is generated. Then the data owner sends $\mathcal{I}$ and $\mathbb{C}$ to the server. The server can use the search token $t(w)$ received from the user to compute the pointers to the documents that the user needs, and returns the corresponding documents to the user. At last, the user uses the private key to decrypt them locally.

 In fact, in order to reduce the space complexity, Song's scheme \cite{song2000practical} did not employ the index structure, so the search complexity is linear in the length of the document. However, they pointed that to improve the search efficiency, it is feasible to reduce the level of security properly. That is to say, a SSE scheme is secure as long as it meets the following demands:
 \begin{enumerate}
   \item The server \emph{cannot learn anything} about the plain documents when it only gets the ciphertexts;
   \item When the server executes search algorithm, it also \emph{cannot} learn anything about the plain documents and the plain keywords querying \emph{except} the search results.
 \end{enumerate}
  Undoubtedly, a stronger privacy guarantee about SSE can be achieved by using the oblivious RAMs \cite{goldreich1996software} or the private information retrieval (PIR) technology \cite{Ishai2006Cryptography}. But they need multiple interactions between the server and the user, which are unpractical.

The cloud storage used above is provided privately. Though it can guarantee the privacy and the security of data, it limits the usage of data. For example, a medical researcher wants to observe the symptoms of the patients infected with human immunodeficiency virus (HIV) to further provide a possible treatment. Therefore, he needs to get access to a large of electronic medical records $(EMRs)$ from different clouds. However, the researcher usually chooses the clouds with low service charges, which may affect the experiment results. For another example, because the clouds do not share their data to each other, the doctors cannot make more effective treatment based on the previous diagnostic results from another hospital. Although cloud storage provides convenience to users, it allows the cloud storage provider to spend more money and energy to maintain these data. Therefore, this cloud storage device cannot be called a real cloud.

Moreover, the server in SSE usually is credible, who always returns the result to the user, even if the result is not correct. That is to say, if the server is a malicious adversary, he may return a wrong result to the user. Though it can utilize the message authentication codes (MAC) technology to resist such adversary, there is a precondition: The server must return something to the user. Therefore, this adversary cannot be called as a malicious adversary. There does not exist any effective solutions to resist a real malicious adversary. However, we can find an alternative: If the server returns wrong results, it cannot charge the service fee. The user can choose someone else to help it to search.
\subsection{Related Work}
Though Song et al. \cite{song2000practical} pointed out that it can introduce some auxiliary information to improve the search efficiency, they did not provide a solution. Goh et al. firstly used Bloom Filter to construct an Index for each document \cite{goh2003secure}. They thought a SSE scheme was secure if it was indistinguishability against chosen keyword attacks (IND-CKA). However, this definition was valid only when the users performed all the searches at once. In addition, it did not require the trapdoor to be secure, which caused the Index unsafe. Therefore, Curtmola et al.\cite{curtmola2006searchable} redefined the security definition and gave two feasible schemes with $\textrm{O}(D(w))$ complexity, where $D(w)$ denotes the number of the documents that contain keyword $w$. This is the first solution to achieve sub-linear complexity.


 Alderman et al. presented a SSE scheme supporting multi-level access policy \cite{aldermanmulti}. The Index was different from that in \cite{curtmola2006searchable}. Namely, the items in the Index $I$ are ordered according to the access level. The documents with higher level will be placed in the front, and those with lower level are put in the end. When retrieving, if the user has lower permission, he only gets the search token with lower grading. However, their construction cannot guarantee the privacy of the users.

Golle et al. firstly considered the conjunction operation on keywords, and gave two solutions\cite{golle2004secure}. The first construction is based on \emph{Decisional Diffie-Hellman} ($\mathbb{DDH}$) assumption, and the search complexity is linear in the number of documents stored on the server. The second protocol is based on \emph{Bilinear Decisional Diffie-Hellman} ($\mathbb{BDDH}$) assumption, whose search complexity is linear in the number of keyword fields. The subsequent work \cite{moataz2013boolean} is also linear in the number of documents, but their scheme can support more general model-\emph{boolean} query.

Cash et al. were the first to reduce the search complexity of boolean expression on keywords into sub-linear \cite{cash2013highly}. Namely, when retrieving the documents that contain keywords
$w_{1},\ldots,w_{ n}$, the search complexity is only linear in the size of the smallest set $DB(w_{ i}) (1\leq i \leq n)$, where $DB(w_{ i})$ denotes the documents that contain keyword $w_{ i}$. However, their protocol only efficiently support such form: $w_{1}\wedge \phi(w_{2},\ldots,w_{ n})$ where $\phi$ is a Boolean formula. Kamara et al. addressed the disjunctive expression $w_{1}\vee w_{2}\vee \ldots\vee w_{ n}$ issue, which can reduce search complexity of the arbitrary Boolean expression into sub-linear in the worst case\cite{kamara2017boolean}. However, because they used the set theoretic terms, the server can compute all the document set $DB(w_{ i}) (i=1,\ldots,n)$. That is to say, it leaks more information than that in Cash's scheme \cite{cash2013highly}.

 Li et al. proposed two methods to solve the problem of fuzzy search\cite{Li2010Fuzzy}, which needed multiple communications. Boldyreva et al. firstly gave the security definition for fuzzy SSE scheme, whose search complexity is sub-linear\cite{boldyreva2014efficient}. Wong et al. used asymmetric scalar product preserving encryption technology to solve the problem of $k-$nearest neighbor (kNN) computation on encrypted database \cite{wong2009secure}. Cao et al. put forward the multi-keyword ranked search algorithm by using kNN idea \cite{cao2014privacy}. Fu et al. designed a central keyword semantic extension ranked scheme \cite{FuHRWW17}.

Kamara et al. proposed a parallel search scheme \cite{kamara2013parallel} which needed $O(r)$ parallel time when querying keyword $w$, where $r$ denoted the number of documents containing keyword $w$. Stefanov et al. firstly solved the problem of forward privacy for dynamic SSE scheme \cite{stefanov2014practical} by using hierarchical structure. Bost et al. pointed out that the scheme \cite{stefanov2014practical} was insecure, and gave the improved schemes \cite{bost2016verifiable,bost2016ovarphiovarsigma}. Van Liesdonk et al. solved the problem of how to dynamically updating the Index and the documents \cite{van2010computationally}. The subsequent works are \cite{kamara2012dynamic,cash2014dynamic,naveed2014dynamic,guo2017dynamic}.

B{\"o}sch et al. made a whole survey of SSE protocols \cite{bosch2015survey}. In order to break the link among the access pattern, the search pattern and the size pattern, Cui et al. suggested that in addition to introducing some dummy data, it should re-randomize and shuffle the physical location of the searched data after executing each query \cite{cui2017secure}.

The main adversary considered in SSE is honest-but-curious \cite{song2000practical,curtmola2006searchable,kamara2013parallel,cash2013highly,kamara2012dynamic,cash2014dynamic,naveed2014dynamic}. Kurosawa et al. firstly used the Message Authentication Code (MAC) technology to resist malicious adversary \cite{kurosawa2012uc}. Cheng et al. utilized indistinguishability obfuscation ($\mathcal{IO}$) against the malicious adversary \cite{cheng2015verifiable}, which can also resist the malicious user. Other works that resist malicious adversary are \cite{bost2016verifiable,bost2016ovarphiovarsigma,Kurosawa2013How}. Dai et al. made use of the physically unclonable function (PUF) to resist the memory attack \cite{Dai2016Memory}. Li et al. introduced the coercer into searchable symmetric encryption \cite{Li2017Deniable}.

\textbf{Bitcoin} is an emerging electronic digital currency in the peer-to-peer (P2P) network. It was firstly proposed by Satoshi Nakamoto \cite{nakamoto2008bitcoin}. The first bucket of Bitcoin was issued in 2009. According to the original assumption, there is only 21 million Bitcoin which will entirely come into the market in 2040. The generation of Bitcoin does not depend on the trusted entity, everybody (i.e. miner) in Bitcoin system may issue a certain amount of Bitcoin as long as he mines a right nonce which is got approval by the majority of nodes.

In order to support the audit, it demands the ledger to be public. Namely, the transactions are publicly stored on the blockchain. After a transaction was created, it was broadcasted to the blockchain where the miners will use the proof-of-work mechanism to verify it. Once it is accepted, the transaction will be stored on a block which is produced on every ten minutes by using the cryptography technology. The Bitcoin can be seen as a purely decentralized system which requires the majority of nodes in peer-to-peer network to be honest. Comparing to the previous electronic currency \cite{chaum1983blind,chaum1984blind,camenisch2005compact}, Bitcoin can support returning change.

There are many works about blockchain in recent years. Ron et al. made a quantitative analysis for the Bitcoin Transaction \cite{ron2013quantitative}. Vitalik et al. firstly introduced smart contract terminology \cite{lavery1999smart} into bitcoin system, and proposed ethereum \cite{buterin2014next}, which can be seen as a sub-chain of the Bitcoin. Andrychowicz et al. and Bentov et al. respectively introduced the Bitcoin into multiparty computations to solve the fairness problem \cite{andrychowicz2014secure,andrychowicz2014fair,bentov2014use}. In fact, the protocols \cite{andrychowicz2014secure,andrychowicz2014fair,bentov2014use} can be seen as a smart contract, because it introduces commitment algorithm $h(x)$ in the out-script of the transaction. Swan put forward several scenarios that the blockchain can be applied to \cite{swan2015blockchain}, one of them is Blockchain health. It provides a structure to store the health data on the blockchain such that it can be analyzed but remain private. The patients who put their own electronic medical record (EMR) onto the blockchain can obtain a certain amount of healthcoin. In this blockchain system, each researcher, such as doctors, pharmacies, insurance companies, and so on, can get access to these data as long as they have the corresponding private keys. However, they did not give an effective search method.
\subsection{Our Contribution}

Putting the data on an open chain is of significant, because they have potential value in medicine and so on. Meanwhile, this open chain can be seen as a real cloud storage since each one can contribute some parts of their own storage space. However, when the data increases, how to perform search quickly is intractable. Taking the Bitcoin system for example, suppose that Alice wants to read transactions created over a period of time, she has to read the data from the last block to the first block, which means the search efficiency is linear into $O(|T|)$, where $|T|$ denotes the number of transactions stored on the blockchain. Therefore, it is necessary to solve the search problem on the existing blockchains

Moreover, at present, the data usually is stored on a private cloud storage, which may limit the usage of the data. If a researcher wants to retrieve some data, he needs to ask the user where his data is stored firstly. Besides, some cloud storages may charge higher service fees. In addition, in the existing SSE schemes, the server is credible,i.e., it often returns the results to the user,
even if the results are incorrect. Though, it can use message authentication code (MAC) to verify whether the result is right or not,
it ends in failure when the server returns nothing. Therefore, it is necessary to weaken the server's rights.

To solve the above issues, we combine the blockchain with SSE, and the corresponding contributions are as follows:
\begin{itemize}
\item We store the encrypted data onto the blockchain, which is a decentralized system. In order to support retrieving, we use the blockchain to construct a new SSE model, which is called as SSE-using-BC. In order to guarantee the privacy and confidentiality of data, we give its security definition.
\item According to the size of the data, we construct two different schemes, which we prove secure under our security definition.
 \item We implement our scheme in Linux system, and the experimental results show that our scheme is feasible and secure.
 \end{itemize}
\noindent\textbf{Organization.} The remainder of this paper is organized as follows. In section 2, we review some tools and the terminologies that will be used in our construction. In section 3, we define our SSE-using-BC model and list the security it should satisfy. In Section 4 we propose our concrete SSE schemes. The analysis of performance and security are shown in section 5. The last section is conclusion.
\section{Preliminaries}
In this section, we mainly review the definition of negligible function, traditional models of SSE and transaction happened in the Bitcoin system respectively. Then we list some notations that will be used.

\begin{definition}
A function $f$ is negligible if for every polynomial $p(\cdot)$ there exists an Integer $N$ such that for all integer $n>N$ it holds that $f(n)<$ $\frac{1}{p(n)}$.
\end{definition}

\subsection{The model of SSE}
 It involves three participants in SSE: \emph{data owner}, \emph{server} and  \emph{user}. The user and the data owner can be a same person. As shown in Fig 1: Suppose that the data owner has $n$ documents $D_1,D_2,...,D_n$ which need to be stored on a private cloud. He encrypts these documents into ciphertexts $C_{1},C_{2},...,C_{ n}$ and generates a corresponding Index $I$, which will be sent to the cloud. When a user wants to retrieve the documents that contain keyword $w$, he computes the search token $t_w$ by taking the keyword $w$ and key $K$ as input, which will  be sent to the cloud server. The server finds the document identifiers by combining $t_{w}$ with $I$, and returns the corresponding documents $C_{ij}$ to the user. At last, the user decrypts $C_{ij}$ locally.

A SSE scheme is secure if the following properties hold:
\begin{itemize}
\item The server cannot learn anything about the plain documents when it only got the ciphertexts.
 \item Once the server executes search, except the search results, it also cannot learn anything about the plain documents and the potential keyword.
 \end{itemize}

\begin{figure}[ht]
\begin{center}
\includegraphics[width=5cm,height=4cm]{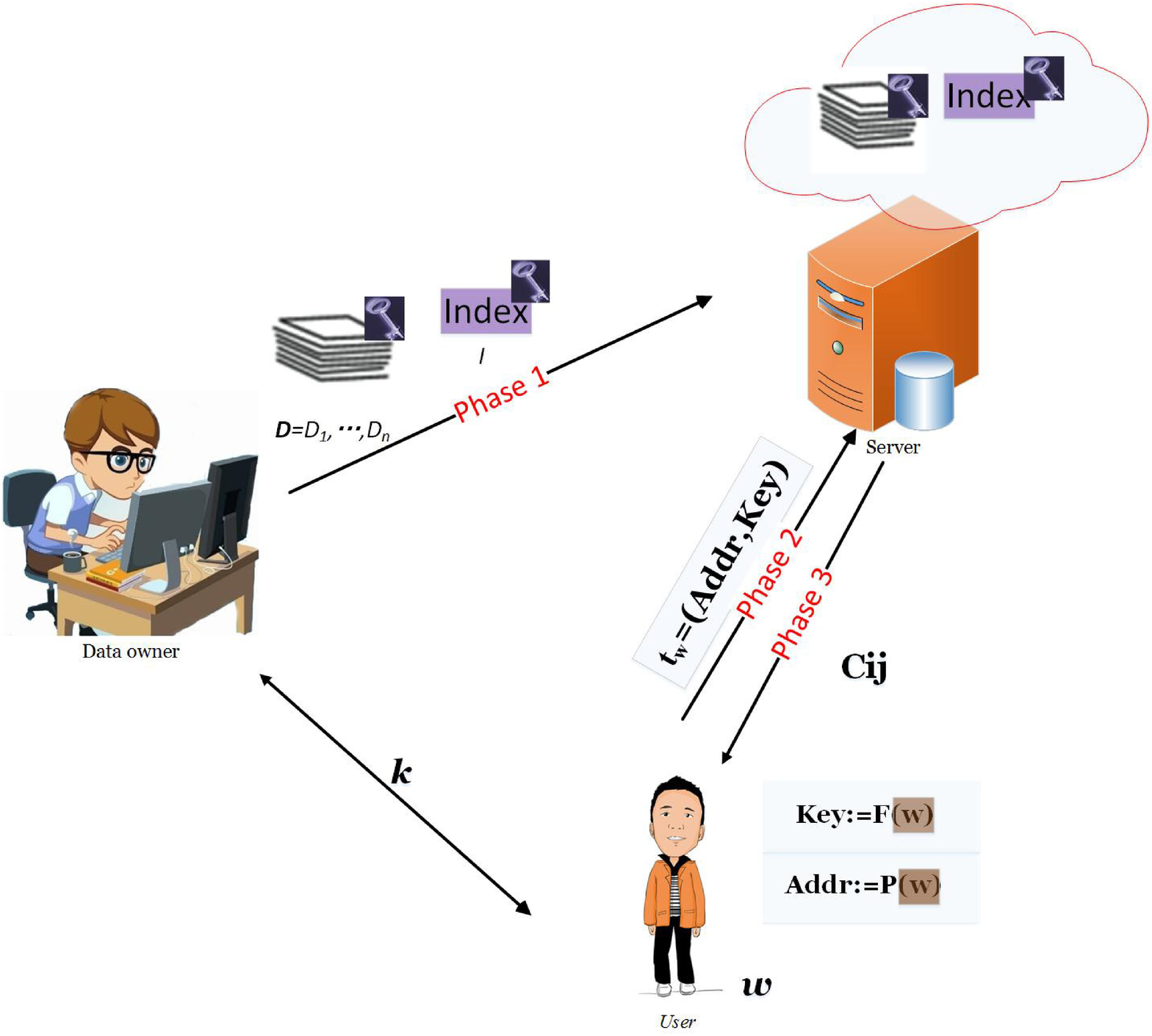}
\caption{SSE Model}
\end{center}
\end{figure}
\subsection{Bitcoin currency system}

The Bitcoin system is composed of addresses and transactions between them. The address usually is a hash value generated by user's public key. Each user can have a pair of keys (i.e., the private key and the public key) when he creates a transaction \cite{andrychowicz2014fair}. The private key is used to sign transactions, while the public key is used to verify whether the signatures $\sigma$ of these transactions are valid or not. For brevity, let we use $(A.pk,A.sk)$ to denote the key pair of the user $A$, and write $\sigma=sig_{A}(m)$ to be the signature of transaction $m$ by using the private key $sk$ of $A$, and $ver_{A}(m,\sigma)$ be the verification by using the public key $pk$ of $A$.

A \emph{transaction} in the Bitcoin system can have multiple inputs and outputs, which describe the circulation of Bitcoin. Let $y_{ i}$ be the hash value of previous transaction $T_{{y_i}}$, $a_{ i}$ be the index of the output of transaction $T_{{y_i}}$. For a transaction, we will use $\sigma_{i}$ to represent its input-script and $\pi_{ i}$ be its output-script, both of them can be written in Bitcoin scripting language, i.e., the stack based language \cite{andrychowicz2014secure}. Therefore, a transaction can be expressed as $T_x=((y_{1},a_{1},\sigma_{1}),$ $...,(y_{ l},a_{l},\sigma_{ l}),$ $(v_{1},\pi_{1}),$$ ...,$ $(v_{ l},\pi_{ l}),t)$, where $(y_{1},a_{1},\sigma_{1}),...,(y_{ l},a_{ l},\sigma_{ l})$ denote the inputs, $(v_{ 1},\pi_{1}),$$ ...,$ $(v_{ l},\pi_{ l}),t) $ denote the outputs, and $v_{ i}$ is the amount of coins. Here, $t$ is a time which is not a compulsory requirement in a transaction. If a transaction includes time t in the out-script, it means that this transaction will be valid only after $t$ time. We will write $[T_x] =(y_{1},a_{1}),...,(y_{ l},a_{ l}),(v_{1},\pi_{1}),...(v_{ l},\pi_{ l}),t)$ to represent the body of $T_{ x}$.

A transaction is valid if and only if it satisfies that:
\emph{(1)} The time $t$ is reached. \emph{(2)} The $\pi_{ i}([T_x],\sigma_{ i})(1\leq i\leq l)$ is valid. \emph{(3)} The involved previous transactions $T_{{y_1}},T_{{y_2}},\ldots,T_{{y_l}}$ were not redeemed. If a transaction is accepted by the nodes on the blockchain, it will be included in one block which is produced about every ten minutes.

As shown in figure 2, it is a transaction $T_{ x }= (y_{1},a_{1},\sigma_{1}, v, \pi_{x},t)$, where the input script is a signature, and the output script is a verification algorithm. We call it as a standard transaction.
\begin{figure}[h]
\begin{center}
\includegraphics[width=6cm,height=3cm]{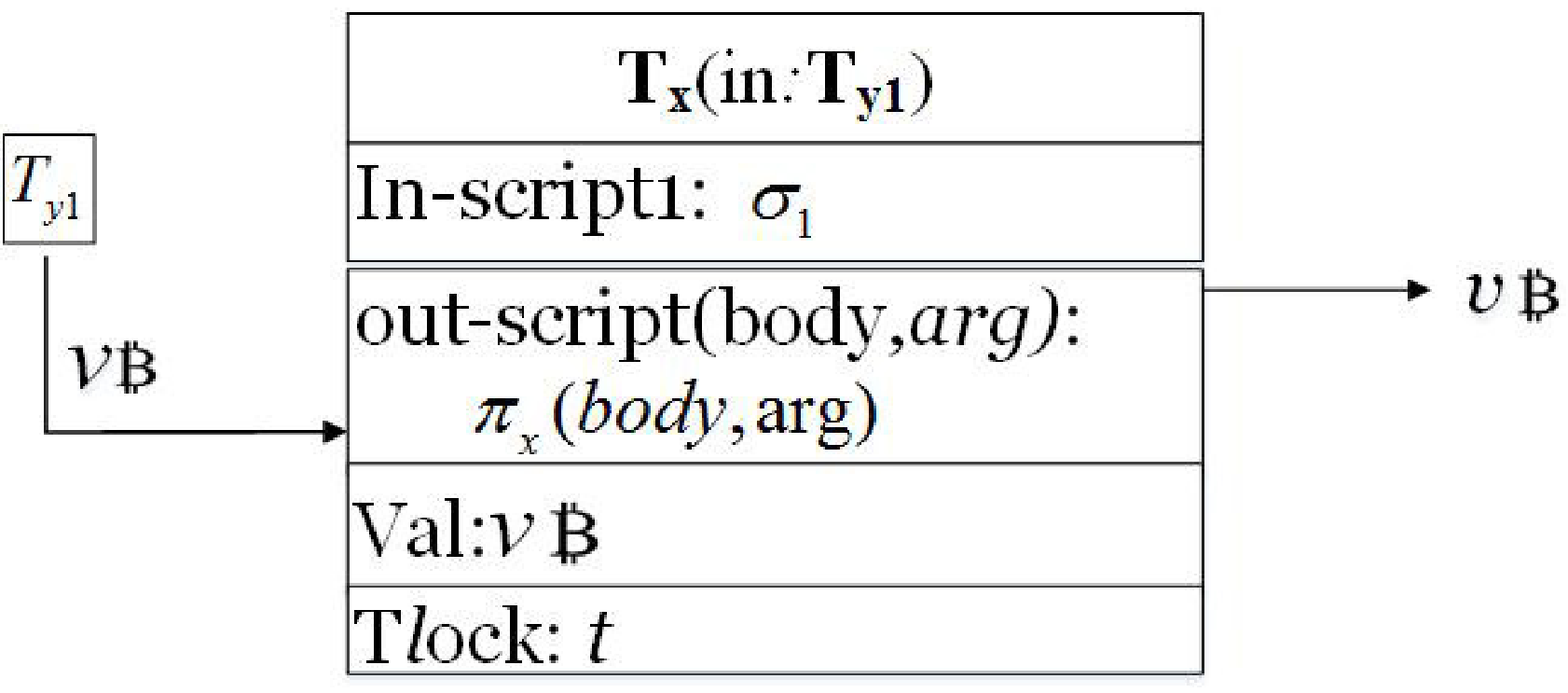}
\caption{transaction $T_x$}
\end{center}
\end{figure}

\vspace{3mm}
Let $x \leftarrow \mathcal{A} (\cdot)$ be the output $x$ of an algorithm $\mathcal{A}$, $x \leftarrow X$ represent an element $x$ sampled uniformly from a set $X$. Let $\varepsilon=(\varepsilon.Enc, \varepsilon.Dec)$ denote a symmetric encryption scheme, where $\varepsilon.Enc$ is the encryption algorithm, and $\varepsilon.Dec$ is the corresponding decryption process. The $a||b$ refers to the concatenation of two string $a$ and $b$. Let $|x|$ represent the length of $x$.
\section{Our System Model}
It is very important to solve the search problem on the blockchain since it is an era of big data nowadays. Take the Bitcoin system for example, each transaction can be seen as a data. When retrieving some transactions, it has to start from the last block until the first block, which means the search efficiency is $O(n)$ where $n$ denotes the number of the data stored on the blockchain. With the increasing number of transactions, this method has become very awkward. Meanwhile, because traditional data often contain the privacy of users, we need to encrypt them before uploading them into the blockchain, which further increases the difficulty to retrieve.

In this section, we firstly build a generic model of SSE-using-BC , then we give its security definition.
\subsection{The model of SSE-using-BC}
As shown in figure 3, it contains the data owner, the user $U^\prime$, the user $Q$ (who is not marked in the Fig.3.) and miners in the model of SSE-using-BC. The data owner has $n$ documents $\mathbb{D}=\{D_{1},D_{2},\ldots,D_{ n}\}$ which need to be uploaded to the blockchain. In order to ensure the privacy and confidentiality of documents, the data owner uses symmetric encryption algorithm to transform them into ciphertexts $C_{1},C_{2},\ldots,C_{ n}$, which will be uploaded on the blockchain in the form of transaction $i (i=1,$ $\ldots,n)$. After they appear on the blockchain, each of them will have a corresponding transaction identifier (TXID). Then, the data owner uses these TXIDs to generate an Index and upload it on the blockchain in the form of transaction $Inx$. The data owner broadcasts the identifier $TX_{Inx}$ of transaction $Inx$ to others. When user $U^\prime$ wants the user $Q$ to help him retrieve the documents that contain keyword $w$, he constructs the transaction $t$ which embeds the information of search token $t(w)$ and $Inx$. If the user $Q$ wants to get the money $d_t$ from transaction $t$, he needs to build transaction $s$ which embeds the information of $C_{ij}$ and the hash $MAC(C_{i1}\|\ldots\|C_{in})$ that the user $U^\prime$ needs. If the transaction $s$ is accepted by the miners, the user $U^\prime$ will get the documents $C_{ij}$. Lastly, he decrypts them locally. If the transaction $s$ does not appear on the blockchain, the user $U^\prime$ will broadcast transaction $p$ to get his $d_t$ dollars back, which is drawn with a dotted line.

\begin{figure}[ht]
\begin{center}
\includegraphics[width=6cm,height=5cm]{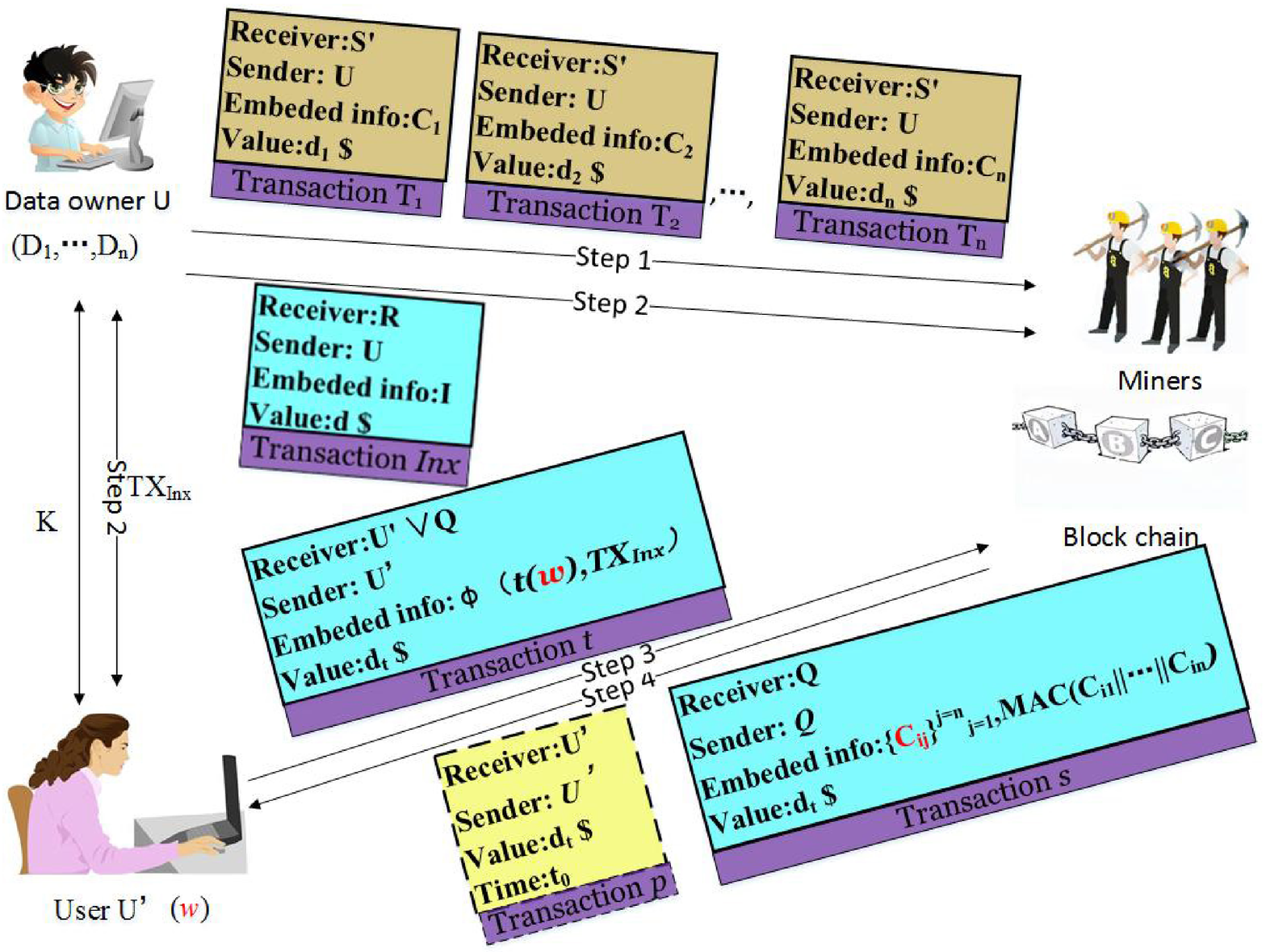}
\caption{The Model of Searchable Symmetric Encryption using BlockChain}
\end{center}
\end{figure}
\textbf{Remark}: The user $Q$, as one of the receivers in the transaction $t$, is not shown in figure 3, because he is not a fixed person. That is to say, when retrieving, the user $U^\prime$ can ask different person to finish it.

 The model of SSE-using-BC is composed of five steps, i.e., $Gen, Enc, Trpdr, Search$ and $Dec$.
\begin{itemize}
  \item $K\leftarrow Gen(1^k):$ is a probabilistic algorithm that is run by the data owner to set up the scheme. It takes a security parameter $k$ as input, and outputs secret key $K$.
      \item $(\{T_j\}_{j=1}^{j=n}, Inx, TX_{Inx})\leftarrow Enc(K, \mathbb{D},\{d_1\$,$ $d_2\$,\ldots,$ $d_n\$,d\$\}):$ is a probabilistic algorithm that is run by the data owner to encrypt documents. It takes the secret key $K$, $\mathbb{D}=\{D_1,\ldots,D_n\}$ and $\{d_1\$,\ldots,$ $ d_n\$,d\$\}$ as input, and output a sequence of transactions $T_1,\ldots,T_n$, an index transaction $Inx$ and $TX_{Inx}$ which is an identifier of transaction $Inx$. The data owner lastly broadcasts $TX_{Inx}$ to legitimate users.
      \item $t\leftarrow Trpdr(K, w, TX_{Inx}, d_t\$)$: is a deterministic algorithm that is run by the user $U^\prime$. It takes as input the secret key $K$, keyword $w$, identifier $TX_{Inx}$ and $d_t\$$, and outputs the transaction $t$.
      \item $(s/p) \leftarrow Search(t)$: is a deterministic algorithm that is run by the user $Q$ and the user $U^\prime$. When the user $Q$ runs this step, it takes the transaction $t$ as input, and outputs transaction $s$. If this step is run by user $U^\prime$, it takes the transaction $t$ as input, and outputs transaction $p$.
  \item $\{D_{ij}\}\leftarrow Dec(K, s)$: is a deterministic algorithm run by the user $U^\prime$ to recover the documents. It takes the secret key $K$ and transaction $s$ as input, and outputs the documents $\{D_{ij}\}$ that he needs.
\end{itemize}

A SSE-using-BC scheme is correct if for all $K\in \mathbb{N}$, for all $K$ output by $Gen(1^k)$, for all $\mathbb{D}\subseteq 2^\Delta$, for all $(\{T_j\}_{j=1}^{j=n}, Inx)$ output by $Enc(K, \mathbb{D},$ $\{d_1,$ $\ldots,$ $ d_n,d\})$, for all $w\in \Delta$,

$Search(Trpdr(K, w, TX_{Inx}, d_t))=s\bigwedge Dec$ $(K,$ $s)$ $=$ $\{D_{ij}\}$, for $1$ $\leq j\leq n$.
\subsection{Security Definition}
A SSE-using-BC scheme is secure if it meets the following demands.
\begin{itemize}
  \item The server \emph{cannot learn anything} about the plain documents when it only gets the ciphertexts;
   \item When a search is finished, \emph{except} the search results, the server also \emph{cannot} learn anything about the plain documents and the plain keywords that the user queried.
  \item If the user $Q$ cannot provide the right documents to the user $U^\prime$ in the transaction $s$, he also cannot get the deposit from the transaction $t$ created by the user $U^\prime$.
\end{itemize}
There are two types of adversary. One is adaptive, the other is non-adaptive. For an adaptive adversary, he can choose a new keyword according to the previous keywords and search results. A non-adaptive adversary must choose the search keywords at one time. Here, we only consider the adaptive adversary. We write $\{UTXO\}$ to denote the set of unredeemed transactions, and $x\leftarrow\{UTXO\}$ to represent some transactions are sampled from the set $\{UTXO\}$ such that the total amount of money in these transactions is $x\$ $. We now present our real/ideal simulation paradigm.
\begin{figure}
\begin{center}
\begin{tabular*}{80mm}{l}
\hline
  $\ \ \ \textbf{Real}^\Pi_{{\mathbb{A}}}(k) \ \ \ $  \\ \hline
  $K\leftarrow \textbf{KeyGen}(1^k)$   \\
  $(\mathcal{D},st_{{\mathcal{A}}})\leftarrow\mathcal{A}_{0}(1^k)$   \\
$(d_1,\ldots,d_n,d)\leftarrow \mathcal{A}_{0}(st_{\mathcal{A}},\{UTXO\})$ \\
 $(T_1,\ldots,T_n,Inx,TX_{Inx})\leftarrow Enc(\mathcal{D},K,d_1,\ldots,d_n,d)$  \\
  $(w_1,st_{{\mathcal{A}}}) \leftarrow\mathcal{A}_{1}(st_{\mathcal{A}},T_1,\ldots,T_n,Inx)$ \\
  $d_{w_1}\leftarrow \mathcal{A}_{1}(st_{\mathcal{A}},\{UTXO\})$ \\
  $t_{w_1}\leftarrow Trpdr(K,w_1,TX_{Inx},d_{w_1})$   \\
  $\{C_{11},\ldots,C_{1n}\}\leftarrow \mathcal{A}_1(st_{\mathcal{A}},T_1,\ldots,T_n,Inx)$\\
  $s_1\leftarrow Search(t_{w_1},\{C_{11},\ldots,C_{1n}\})$ \\
  for $2\leq i\leq q$  \\
  $(w_i,st_{{\mathcal{A}}}) \leftarrow \mathcal{A}_{i}(st_{\mathcal{A}},T_1,\ldots,T_n,Inx,t_{w_1},\ldots,t_{w_{i-1}})$ \\
  $d_{w_i}\leftarrow \mathcal{A}_{i}(st_{\mathcal{A}},\{UTXO\})$ \\
  $t_{w_i}\leftarrow Trpdr(K,w_i,TX_{Inx},d_{w_i})$  \\
  $\{C_{i1},\ldots,C_{in}\}\leftarrow \mathcal{A}_i(st_{\mathcal{A}},T_1,\ldots,T_n,Inx)$\\
  $s_i\leftarrow Search(t_{w_i})$ \\
let $Tr=(t_{w_1},\ldots,t_{w_q})$, $S=(s_1,\ldots,s_q)$ \\
output $V=(Inx,\{T_1,\ldots,T_n\},Tr,S)$ and $st_{{\mathcal{A}}}$ \\
\\\hline
\end{tabular*}
\caption{Game $Real^\Pi_{{\mathbb{A}}}(k)$}
\end{center}
\end{figure}
\begin{figure}
\begin{center}
\begin{tabular*}{80mm}{l}
\hline
  $\ \ \ \textbf{Ideal}^\Pi_{{\mathbb{A},\mathbb{S}}}(k) \ \ \ $ \\ \hline
  $(\mathcal{D},st_{{\mathcal{A}}})\leftarrow\mathcal{A}_{0}(1^k)$ \\
  $\mathbb{D}=(d_1,\ldots,d_n,d)\leftarrow \mathcal{A}_{0}(st_{\mathcal{A}},\{UTXO\})$ \\
  $(T_1, \ldots,T_n,Inx,TX_{Inx},st_{\mathcal S})\leftarrow \mathcal S_0 (L(\mathcal D),\mathbb{D})$\\
$(w_1,st_{{\mathcal{A}}}) \leftarrow\mathcal{A}_{1}(st_{\mathcal{A}},\{T_1, \ldots,T_n\},Inx)$\\
  $d_{w_1}\leftarrow \mathcal{A}_{1}(st_{\mathcal{A}},\{UTXO\})$\\
  $(t_{w_1},st_{\mathcal{S}})\leftarrow \mathcal{S}_1 (st_{\mathcal{S}},L(\mathcal D,w_1,TX_{Inx}),d_{w_1})$\\
  $(\{C_{11},\ldots,C_{1n}\},st_{\mathcal A})\leftarrow \mathcal{A}_1(st_{\mathcal{A}},T_1,\ldots,T_n,Inx)$\\
  $(s_1,st_{\mathcal S})\leftarrow \mathcal S_1 (st_{\mathcal{S}}, t_{w_1}, \{C_{11},\ldots,C_{1n}\})$\\
 for $2\leq i\leq q$\\
  $(w_i,st_{\mathcal A}) \leftarrow \mathcal{A}_{i}(st_{\mathcal{A}},\{T_1, \ldots,T_n\},Inx,t_{w_1},\ldots,t_{w_{i-1}})$\\
  $d_{w_i}\leftarrow \mathcal{A}_{i}(st_{\mathcal{A}},\{UTXO\})$\\
  $(t_{w_i},st_{\mathcal{S}})\leftarrow \mathcal{S}_i (st_{\mathcal{S}},L(\mathcal D,w_1,\ldots,w_i,TX_{Inx}),d_{w_i})$\\
  $(\{C_{i1},\ldots,C_{in}\},st_{\mathcal A})\leftarrow \mathcal{A}_i(st_{\mathcal{A}},T_1,\ldots,T_n,Inx)$\\
  $(s_i,st_{\mathcal S})\leftarrow \mathcal S_i (st_{\mathcal{S}}, t_{w_i})$\\
  let $Tr=(t_{w_1},\ldots,t_{w_q})$, $S=(s_1,\ldots,s_q)$ \\
output $V=(Inx,\{T_1, \ldots,T_n\},Tr,S)$ and $st_{{\mathcal{A}}}$ \\
\\\hline
\end{tabular*}
\caption{Game $Ideal^\Pi_{{\mathbb{A},\mathbb{S}}}(k)$}
\end{center}
\end{figure}
\begin{definition}
Let $\Pi=(Gen, Enc,Trpdr,$ $Search,Dec)$ be a SSE-using-BC scheme, $L$ denote the leakage function which can be parameterized by access pattern, search pattern and size pattern defined in \cite{curtmola2006searchable}, $k$ be the security parameter, $\mathbb{A}=(\mathcal{A}_{0},\mathcal{A}_{1},...,\mathcal{A}_{ q})$ be an adversary where $q\in \mathbb{N}$, and $\mathbb{S}=(\mathcal{S}_{0},\mathcal{S}_{1},...,\mathcal{S}_{ q})$  be a simulator. Considering the games $Real^\Pi_{{\mathbb{A}}}(k)$ and $Ideal^\Pi_{{\mathbb{A},\mathbb{S}}}(k)$ shown in the figure 4 and 5.
\end{definition}
We say a SSE-using-BC scheme is adaptively semantically secure if for all polynomial size adversaries $\mathbb{A}=(\mathcal{A}_{0},\mathcal{A}_{1},...,\mathcal{A}_{ q})$ where $q=poly(k)$, there exists a non-uniform polynomial size simulator $\mathbb{S}=(\mathcal{S}_{ 0},\mathcal{S}_{1},...,\mathcal{S}_{ q})$, such that for all polynomial size $\mathcal D$,\\
$|Pr[\mathcal D(V,st_{{\mathcal{A}}})=1:(V,st_{{\mathcal{A}}})\leftarrow Real^\Pi_{{\mathbb{A}}}(k)]-Pr[\mathcal D(V,st_{{\mathcal{A}}})=1:(V,st_{{\mathcal{A}}})\leftarrow Ideal^\Pi_{{\mathbb{A},\mathbb{S}}}(k)] |$ $\leq neg(k)$\\
where the probabilities are taken over the coins of $Gen$, $Enc$, $Trpdr$ and $Search$ processes.

\section{The detailed scheme}
If users upload their documents on the blockchain, any researcher can get access to them conveniently. In order to make the blockchain growing normally, it usually requires that the size of each block is fixed. However, the scale of data is different, so some data cannot be stored on the blockchain directly. According to this issue, we consider two situations.

\subsection{A SSE-using-BC scheme for the Lightweight Data}
In this case, there are four participants: Data owner, user $U^\prime$, user $Q$ and miners. The data owner will upload $n$ (small integer) lightweight documents onto the blockchain. The miners will collect these documents and store them on the blockchain. The user $U^\prime$ wants to retrieve some documents that he is interested in. The user $Q$ will return the documents to the user $U^\prime$. Here, the data owner and user $U^\prime$ can be same, if they are two different person, they will share the secret key.

Let $F_1, F_2, F_3$ be three pseudorandom functions, where $ F_{1}: \{0,1\}^k\times\{0,1\}^*\rightarrow \{0,1\}^k$, $ F_{2}: \{0,1\}^k\times\{0,1\}^*\rightarrow \{0,1\}^{k}$, $ F_{3}: \{0,1\}^k\times\{0,1\}^*\rightarrow \{0,1\}^{k}$, $\varepsilon=(\varepsilon.Enc, \varepsilon.Dec)$ be an $IND-PCPA$ secure symmetric encryption scheme, $\epsilon=(\epsilon.Enc,\epsilon.Dec)$ be a deterministic symmetric encryption algorithm and $H$ be a hash function of the merkle-Damg{\aa}rd type which maps $l\cdot p-$th strings to $p-$th strings, where $p$ and $l$ are fixed integers. Suppose each transaction identifier of length $p$ can be computed by the transaction itself.

The concrete construction is composed of five steps.
\begin{itemize}
  \item $\mathbf{Gen}$: It takes the security parameter $k$ as input, and outputs a secret key array $\mathbb{K}=(K_{1},K_{2})$, where $K_{ i}\leftarrow \{0,1\}^k (1\leq i\leq 2)$.

  \item $\mathbf{Enc}$: Firstly, the data owner will use the private key $K_1$ to transform the documents $\mathbb{D}=\{D_{1},D_{2},\ldots,D_{ n}\}$ into $\mathbb{C}=\{C_{1},C_{2},\ldots,C_{ n}\}$:
      \begin{center}
  $C_{ i}=\varepsilon.Enc(K_{ 1},D_{ i})(1\leq i\leq n)$,
   \end{center}
       To store the ciphertext $C_{ i} (1\leq i \leq n)$, he finds $n$ unredeemed transactions $TX_{{D0i}} (i=1,\dots,n) $ of value $d_{i} \$$ whose receiver is data owner, and builds the following transactions $TX_{{D_i}} (i=1,\dots,n)$:
      \begin{itemize}
        \item He computes the body of transaction $TX_{{D_i}}$ by using the $TX_{{D0i}}$ as input.
        \item He embeds $C_{ i} (1\leq i \leq n)$ into the out-script of transaction $TX_{{D_i}}$. After signing it, he broadcasts it to the blockchain.
        \item The miners will collect these transactions. If the transaction $TX_{{D_i}} (i=1,\dots,n)$ appears on the ledger, the data owner records its corresponding transaction identifier $TXID_{{D_i}}$.
      \end{itemize}

      Let $\mathcal {W}=\{w_{1},w_{2},\ldots,w_{ m}\}$ denote a dictionary composed of keywords that appear in the $\mathbb{D}$, where $m$ means the number of keywords. For each keyword $w_{i}(1\leq i\leq m)$, he selects the set $DB(w_{i})$ that is initialized to be empty. If the document $D_{{i_{ j}}}$ contains keyword $w_{i}$, he then puts $TXID_{{D_{{i_{ j}}}}}$ into $DB(w_{i})$. Suppose $\Delta_i$ represents the number of elements in $DB(w_{i})$, and let
      $\Delta=\max\limits_{1\leq i\leq m}\{\Delta_i\}$.
       If $\Delta_i<\Delta$, the data owner pads the remaining $\Delta-\Delta_i$ elements with $0^p$, such that the size of $DB(w_{i})$ is equal to $\Delta$.

Now, the data generates the following transaction which can be seen the documents' Index.
\begin{itemize}
  \item For each keyword $w_i\in \mathcal{W}$, he firstly computes:
  \begin{center}
   $t_{w_i}=F_1(K_2,w_i)$, $l_{w_i}=F_2(K_2,w_i)$,\\
   $e_{w_i}=\epsilon.Enc(l_{w_i},DB(w_i))$, $k_{w_i}=F_3(K_2, w_i)$,\\
   $h_{w_i}=H(k_{w_i},C_{i1}\|\ldots \|C_{in})$.
  \end{center}
   He then puts $(t_{w_i},e_{w_i},h_{w_i})$ into the array $\mathcal I$ in the order of dictionary.
  \item The data owner finds unredeem transaction $TX_0$ of value $d_0 \$$, whose receiver is himself.
  \item He uses the transaction transaction $TX_0$ to compute the body of transaction $Inx$. He then embeds $I$ in the out-script of $Inx$ and signs it. Next, he broadcasts it on the blockchain.
  \item If the transaction $Inx$ appears on the blockchain, he broadcasts its identifier $TX_{Inx}$ to others.
\end{itemize}

\item $\mathbf{Trpdr}$: Let $\phi(,)$ be a function that composed of decryption algorithm and verification algorithm. when inputting $x,y$, this function firstly uses $y$ to locate the corresponding transaction $q$, and uses $x$ to decrypt the information embedded in $q$. Suppose the result is $\alpha,\beta$. It then takes the $\alpha,\beta,x$ as input and verifies if $\beta\stackrel{?}{=}H(x,\alpha)$. If it holds, this function outputs $\alpha,1$.

    Now, the user $U^\prime$ wants to find the segment of documents that contain keyword $w$. As shown in figure 6, he will create the transaction $ask$, the concrete process is as follows:
     \begin{itemize}
  \item Appoint a person to search, suppose it is the user $Q$.
  \item Find an unredeemed transaction $T_q$ of value  $d_t \$$, whose receiver is the user $U^\prime$. He then uses $T_q$ to compute the body of $ask$.
  \item Compute $t_w=F_1(K_2,w)$, $l_w=F_2(K_2,w)$ and $k_w=F_3(K_2, w)$.
  \item Both $U^\prime$ and $Q$ use transaction $ask$ to compute the body of transaction $Fuse$ with time lock set to some time $t$ in the future. $Q$ signs the transaction $Fuse$ and sends it to $U^\prime$. The user $U^\prime$ puts his signature on it.
  \item The user $U^\prime$ puts $((t_w,l_w,k_w),$ $TX_{Inx})$ into the out-script of $ask$.
  \item After signing the transaction $ask$, he broadcasts it.
  \item If the transaction $ask$ does not appear on the blockchain until time $t-max_{U^\prime}$, where $max_{U^\prime}$ is the maximal possible delay of including it in the blockchain, the user $U^\prime$ immediately redeems the transaction $T_q$ by using his private key and quits the protocol.
\end{itemize}
\begin{figure}[ht]
 \centering
\includegraphics[width=7cm,height=7cm]{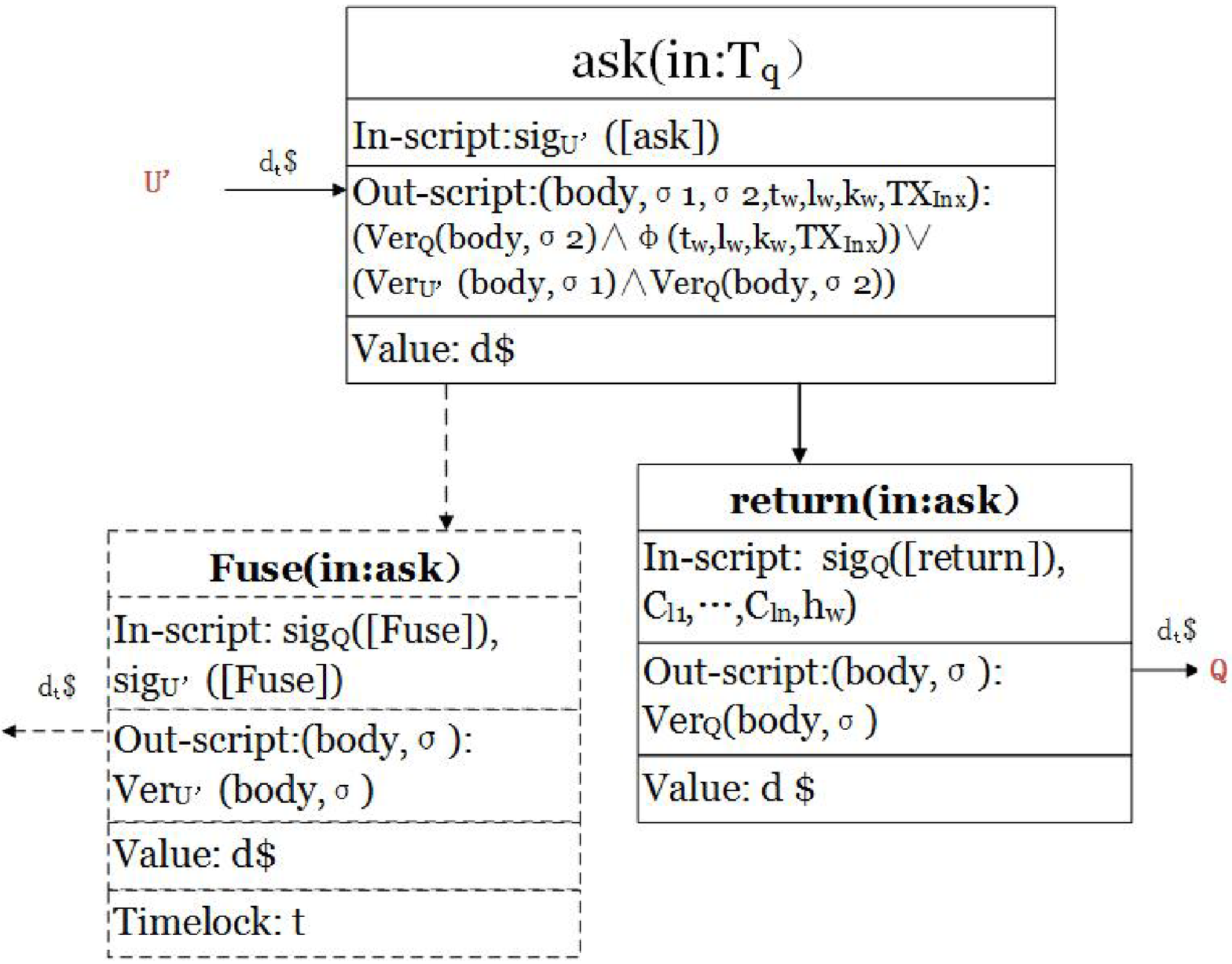}
\caption{get the documents that contain keyword $w$}
\end{figure}
   \item $\mathbf{Search}$: If the user $Q$ wants to get the money from the transaction $ask$, he must compute the result by running the function $\phi(,)$ and creates the transaction $return$.
       \begin{itemize}
       \item Compute the body of $return$ transaction by using $ask$ transaction as input.
         \item Run the function $\phi((t_w,l_w,k_w),TX_{Inx})$. Namely, he firstly uses $TX_{Inx}$ to read the information $I$ embedded in the transaction $Inx$. Then, he finds $(e_w,h_w)$ from $I$ by using $t_w$. Next, the user $Q$ decrypts $e_w$: $DB(w)=\epsilon.Dec(l_{w},e_w)$. Suppose $DB(w)=\{TXID_{D_{l_1}},TXID_{D_{l_2}},$ $\dots,TXID_{D_{l_n}}\}$, where $TXID_{D_{l_j}} (j=1,\ldots,n)$ is the identifier of transaction $TX_{D_{l_j}} (j=1,\ldots,n)$. The user $Q$ reads the document ciphertext $C_{l_j}$ from transaction $TX_{D_{l_j}}$.
             \item Embed the $(\{C_{l_j}\}, h_w)$ into the out-script of transaction $return$.
             \item Broadcast transaction $return$ to the blockchain with his signature on it.
       \end{itemize}
   \item $\mathbf{Dec}$: If the transaction $return$ appears on the blockchain, the user $Q$ can get the $\{C_{l_j}\}$ from it. Then he uses private key to compute $D_{l_j}=\varepsilon.Dec(K_1,C_{l_j})$ $ (1\leq j\leq n)$. If within time $t$ the transaction $return$
does not appear on the blockchain, the user $Q$ broadcasts transaction $Fuse$ and get his money back.

\end{itemize}
\subsection{A SSE-using-BC scheme for the Big Data}
It indirectly indicates that the size of the data and the index is not large in the above case. However, there always exists big data in reality, which will be rejected in the blockchain system since its size is larger than the maximum value of the transaction that the system allows to be. Therefore, we needs to process the data before storing them on the blockchain. In this section, we give a feasible scheme to solve this issue.

Before the scheme beginning, the user will choose the following functions. Let $F_1, F_2, F_3$ be three pseudorandom functions, where $ F_{1}: \{0,1\}^k\times\{0,1\}^*\rightarrow \{0,1\}^k$, $ F_{2}: \{0,1\}^k\times\{0,1\}^*\rightarrow \{0,1\}^{k}$, $ F_{3}: \{0,1\}^k\times\{0,1\}^*\rightarrow \{0,1\}^{k}$, $\varepsilon=(\varepsilon.Enc, \varepsilon.Dec)$ be an $IND-PCPA$ secure symmetric encryption scheme, $\epsilon=(\epsilon.Enc,\epsilon.Dec)$ be a deterministic symmetric encryption algorithm and $H$ be a keyed hash function $H:\{0,1\}^k\times \{0,1\}^*\rightarrow \{0,1\}^k$. Suppose the length of the transaction cannot exceed $\iota$, and the length of identifier of the transaction is $p$.

The concrete scheme is composed of five steps which is listed as follows.

\begin{itemize}
  \item $\mathbf{Gen}$: It takes the security parameter $k$ as input, and outputs a secret key array $\mathbb{K}=(K_{1},K_{2})$, where $K_{ i}\leftarrow \{0,1\}^k (1\leq i\leq 2)$.

  \item $\mathbf{Enc}$: Firstly, the data owner will use the private key $K_1$ to transform the documents $\mathbb{D}=\{D_{1},D_{2},\ldots,D_{ n}\}$ into $\mathbb{C}=\{C_{1},C_{2},\ldots,C_{ n}\}$:
  \begin{center}
  $C_{ i}=\varepsilon.Enc(K_{ 1},D_{ i})(1\leq i\leq n)$.
   \end{center}
       \begin{itemize}
        \item If $|C_i|>\iota$, the data owner will divide $C_i$ into $s$ blocks $\tilde{C}_{i1},\tilde{C}_{i2},\ldots,\tilde{C}_{is}$, such that $|\tilde{C}_{ij}|+p\leq\iota, \forall j\in \{1,\ldots,s\}$, where $s=\lceil\frac{|C_i|}{\iota-p}\rceil$. To store $C_i$, he finds $s$ unredeemed transactions $TX_{{\tilde{D}0_{ik}}} (k=1,\dots,s)$ of value $d_{ik}\$$, whose receiver is the data owner. He then builds transactions $TX_{{\tilde{D}_{ik}}} (k=1,\dots,s)$ shown as follows:
              \begin{itemize}
               \item When $k=1$:
                  \begin{itemize}
                    \item Compute the body of transaction $TX_{{\tilde{D}_{i1}}}$ by using the $TX_{{\tilde{D}0_{i1}}}$ as input.
                    \item Embed $\tilde{C}_{i1}\|0^p$ into the out-script of transaction $TX_{{\tilde{D}_{i1}}}$. After singing it, he broadcasts it to the blockchain.
                    \item If the $TX_{{\tilde{D}_{i1}}}$ appears on the blockchain, he records its identifier $TXID_{\tilde{D}_{i1}}$.
                    \end{itemize}
               \item For $2\leq k\leq s$:
                  \begin{itemize}
                   \item Compute the body of transaction $TX_{{\tilde{D}_{ik}}}$ by using the $TX_{{\tilde{D}0_{ik}}}$ as input.
                   \item Embed $\tilde{C}_{ik}\|TXID_{\tilde{D}_{i(k-1)}}$ into the out-script of transaction $TX_{{\tilde{D}_{ik}}}$. After singing it, he broadcasts it to the blockchain.
                   \item If the transaction $TX_{{\tilde{D}_{ik}}}$ appears on the ledger, he records its corresponding transaction identifier $TXID_{{\tilde{D}_{ik}}}$.
                   \end{itemize}
             \end{itemize}
     \item If $|C_i|\leq\iota (1\leq i \leq n)$, he then finds an unredeemed transaction $TX_{{D0i}}$ of value $d_{i}\$$, whose receiver is data owner. He next builds the following transaction $TX_{{D_i}}$:
           \begin{itemize}
            \item Computes the body of transaction $TX_{{D_i}}$ by using the $TX_{{D0i}}$ as input.
           \item Embeds $C_{ i}$ into the out-script of transaction $TX_{{Di}}$. After signing it, he broadcasts it to the nodes.
          \item The miners will collect it. If the transaction $TX_{{D_i}}$ appears on the ledger, he records its corresponding transaction identifier $TXID_{{D_i}}$.
           \end{itemize}
        \end{itemize}
      Let $\mathcal {W}=\{w_{1},w_{2},\ldots,w_{m}\}$ denote the dictionary composed of keywords that appear in the $\mathbb{D}$, where $m$ is the number of keywords. For each keyword $w_{i}(1\leq i\leq m)$, he selects the set $DB(w_{i})$ that is initialized to be empty.
      \begin{itemize}
        \item If $w_i\in D_{i_j}\bigwedge|C_{i_j}|>\iota$, he puts $TXID_{\tilde{D}_{i_j}s}$ into the set $DB(w_{i})$.
        \item If $w_i\in D_{i_j}\bigwedge|C_{i_j}|\leq\iota$, he puts $TXID_{D_{i_j}}$ into the set $DB(w_{i})$.
      \end{itemize}
 Suppose $\Delta_i$ represents the number of elements in $DB(w_{i})$, let $\Delta=\max\limits_{1\leq i\leq m}\{\Delta_i\}$. If $\Delta_i<\Delta$, the data owner pads $0^p$ for the remaining $\Delta-\Delta_i$ elements in $DB(w_{i})$.

Now, the data owner generates the following transaction for each keyword $w_i\in \mathcal{W}$. Namely, he firstly computes:
     \begin{center}
   $t_{w_i}=F_1(K_2,w_i)$, $l_{w_i}=F_2(K_2,w_i)$,\\
   $e_{w_i}=\epsilon.Enc(l_{w_i},DB(w_i))$, $k_{w_i}=F_3(K_2, w_i)$,\\
   $h_{w_i}=H(k_{w_i},C_{i1}\|\ldots \|C_{in})$.
     \end{center}
       \begin{itemize}
       \item In order to create transaction $TX_{Iw_1}$ for keyword $w_1$, the data owner finds an unredeem transaction $TX_{Iw10}$ of value $d_{w10} \$$, whose receiver is himself.
       \item He uses the transaction transaction $TX_{Iw10}$ to compute the body of transaction $TX_{Iw_1}$. He then computes
       \begin{center}
       $K_{11}=F_2(K_2,0^p)$, $r_1=\epsilon.Enc(K_{11},t_{w_1}\|e_{w_1}\|h_{w_1}\|0^p)$,
       \end{center}
      and embeds $r_1$ in the out-script of $TX_{Iw_1}$ and signs it. At last, he broadcasts it on the blockchain.
       \item If the transaction $TX_{Iw_1}$ appears on the blockchain, he records its identifier $TI_{w_1}$, which can be seen as a pointer to $TX_{Iw_1}$.
      \item If transaction $TX_{Iw_1}$ does not appear on the blockchain, the data owner can redeem transaction $TX_{Iw10}$ quickly and quits the protocol.
       \end{itemize}
     For $w_j (2\leq j\leq m)$, the data owner creates transaction $TX_{Iw_j}$ which is shown as follows:
       \begin{itemize}
      \item The data owner finds an unredeem transaction $TX_{Ij0}$ of value $d_{j0} \$$, whose receiver is himself.
      \item He uses the transaction transaction $TX_{Ij0}$ to compute the body of transaction $TX_{Iw_j}$. He computes
       \begin{center}
       $K_{11}=F_2(K_2,0^p)$, $r_j=\epsilon.Enc(K_{11},t_{w_j}\|e_{w_j}\|h_{w_j}\|TI_{w_{j-1}})$,
       \end{center}
  then embeds $r_j$ in the out-script of $TX_{Iw_j}$ and signs it. At last, he broadcasts it on the blockchain.
      \item If the transaction $TX_{Iw_j}$ appears on the blockchain, he records its identifier $TI_{w_j}$, which can be seen as a pointer to $TX_{Iw_j}$.
      \item If transaction $TX_{Iw_j}$ does not appear on the blockchain, the data owner can redeem transaction $TX_{Ij0}$ quickly and quits the protocol.
      \end{itemize}
  At last, the data owner broadcasts $TI_{w_m}$ to others who have the permission of search.

\item $\mathbf{Trpdr}$: Let $\phi(\cdot,\cdot)$ be a function that composed of decryption algorithm and verification algorithm. when inputting $x,y$, this function firstly uses $y$ to locate the corresponding transaction $q$, and uses $x$ to decrypt the information embedded in $q$. Suppose the result is $\alpha,\beta$. It then takes $\alpha,\beta,x$ as input and verifies if $\beta\stackrel{?}{=}H(x,\alpha)$. If it holds, this function outputs $\alpha,1$.

    Now, the user $U^\prime$ wants to find the segment of documents that contain keyword $w$. As shown in figure 7, he will create the transaction $ask$, the concrete process is as follows:
     \begin{itemize}
  \item Appoint a person to search, suppose it is the user $Q$.
  \item Find an unredeemed transaction $T_q$ of value $d_t \$$, whose receiver is the user $U^\prime$. He then uses $T_q$ to compute the body of $ask$.
  \item Compute $t_w=F_1(K_2,w)$, $l_w=F_2(K_2,w)$, $K_{11}=F_2(K_2,0^p)$ and $k_w=F_3(K_2, w)$.
  \item Both $U^\prime$ and $Q$ use transaction $ask$ to compute the body of $Fuse$ transaction with time lock set to some time $t$ in the future. $Q$ sends $Fuse$ to $U^\prime$ with his signature on it. Then, the user $U^\prime$ can put his signature on it.
  \item The user $U^\prime$ embeds $((t_w,l_w,k_w),$ $K_{11},$ $TI_{w_m})$ into the out-script of $ask$.
  \item After signing the transaction $ask$, he broadcasts it.
  \item If the transaction $ask$ does not appear on the blockchain until time $t-max_{U^\prime}$, where $max_{U^\prime}$ is the maximal possible delay of including it in the blockchain, the user $U^\prime$ immediately redeems the transaction $T_q$ by using his private key and quits the protocol.
\end{itemize}
\begin{figure}[ht]
 \centering
\includegraphics[width=7cm,height=5.5cm]{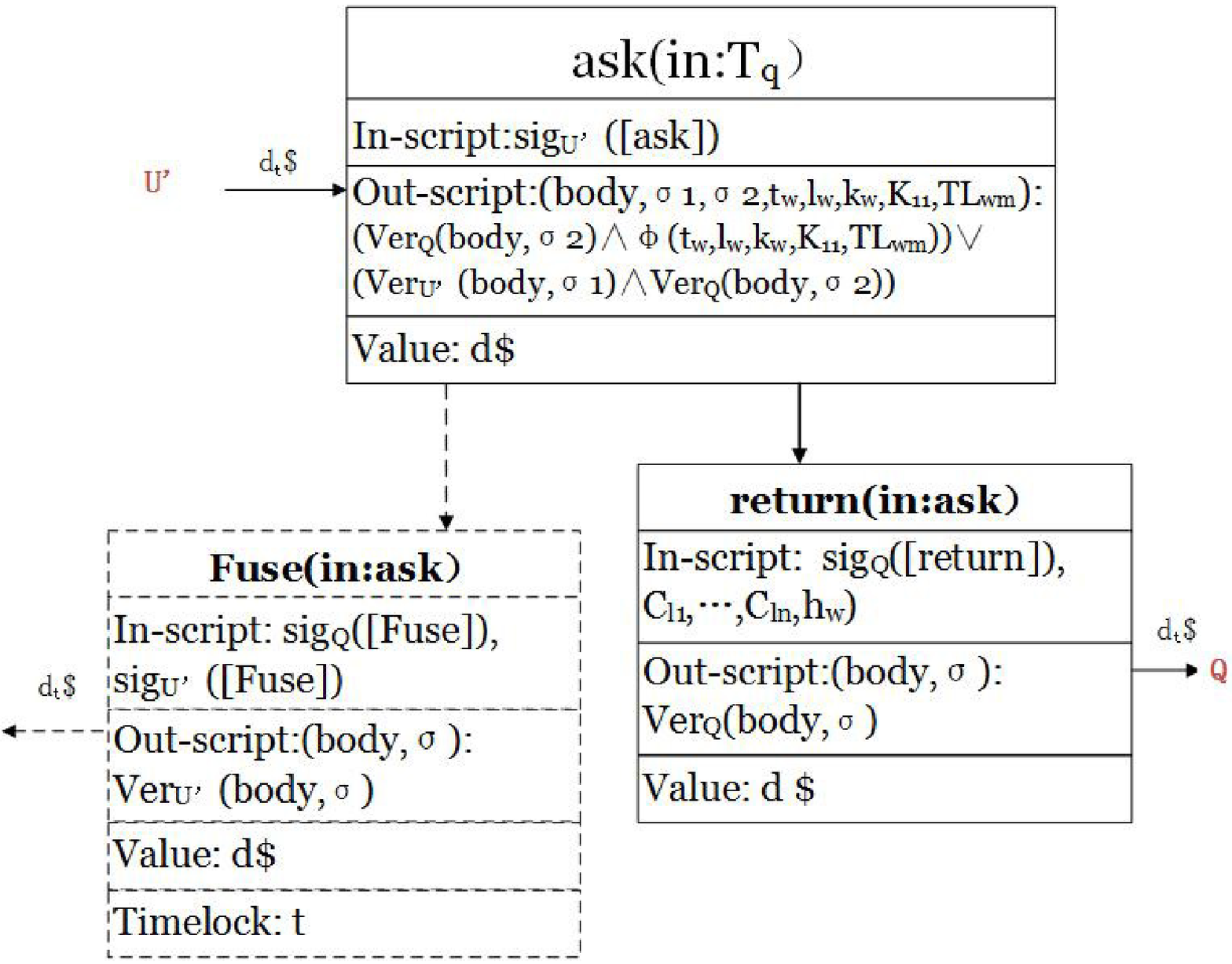}
\caption{return the documents that contain keyword $w$}
\end{figure}
   \item $\mathbf{Search}$: As shown in figure 7, if the user $Q$ wants to claim the money from the transaction $ask$, he must do:
       \begin{itemize}
       \item Compute the body of transaction $return$ transaction by using transaction $ask$ as input.
         \item Run the function $\phi(t_w,l_w,k_w,K_{11},TI_{w_m})$. Firstly, he uses $TI_{w_m}$ to read the information $r_m$ embedded in the transaction $TX_{Iw_m}$, and computes
              $t_{w_m}\|e_{w_m}\|h_{w_m}\|TI_{w_{m-1}}=\epsilon.Dec(K_{11},r_m)$.
              \begin{itemize}
                \item If $t_{w_m}=t_w$, he continues to do $DB_{w_m}=\epsilon.Dec(l_{w_m},e_{w_m})$. Suppose $DB_{w_m}=\{TXID_{D_{m_1}},$ $\ldots,TXID_{D_{m_\Delta}}\}$. He reads ciphertext $C_i$ by using $TXID_{D_{m_i}}(1\leq i\leq \Delta)$:
                    \begin{itemize}
                      \item If the information embedded in the transaction $TX_{D_{m_i}}$ is $C_{m_i}$, he records this value.
                      \item If the information embedded in the transaction $TX_{D_{m_i}}$ is $\tilde{C}_{{m_i}s}\|$ $TXID_{\tilde{D}_{{m_i}(s-1)}}$, he firstly records $\tilde{C}_{{m_i}s}$, and continues to use the transaction identifier $TXID_{\tilde{D}_{{m_i}j}}$ to
                          read the information $\tilde{C}_{{m_i}j} (j=s-1,\ldots,1)$ embedded in the transaction $TX_{\tilde{D}_{{m_i}j}}(j=s-1,\ldots,1)$. Then, he computes $C_{m_i}=\tilde{C}_{{m_i}1}\|\ldots\|\tilde{C}_{{m_i}s}$.
                    \end{itemize}
                \item If $t_{w_m}\neq t_w$, he uses transaction identifier $TI_{w_{m-j}}$ to read information $r_{m-j} (j=1,\ldots,m-1)$ embedded in the transaction $TX_{Iw_{m-j}}(j=1,\ldots,m-1)$ until it stops. That is to say, he does:
                    \begin{itemize}
                      \item Decrypt $t_{w_{m-j}}\|e_{w_{m-j}}\|h_{w_{m-j}}\|TI_{w_{{m-j}-1}}=\epsilon.Enc(K_{11},r_{m-j})$,
                      \item Verify $t_{w_{m-j}}\stackrel{?}{=}t_w$. If it holds, he decrypts $DB_{w_{m-j}}$ and gets $\{C_{l1},\ldots,C_{ln}\}$ by using the above method. If it does not hold, he continues to read the information $r_{m-j-1}$ embedded in the transaction $TX_{Iw_{m-j-1}}$.
                   \end{itemize}
                \end{itemize}
             \item Embed the $(\{C_{l_1},\ldots,C_{l_n}\}, h_w)$ into the out-script of transaction $return$.
             \item Sign the transaction $return$ and broadcast it.
       \end{itemize}
   \item $\mathbf{Dec}$: If the transaction $return$ appears on the blockchain, the user $Q$ can get the $\{C_{l_j}\}$ from it. Then he uses private key to compute $D_{l_j}=\varepsilon.Dec(K_1,C_{l_j})$ $ (1\leq j\leq n)$. If within time $t$ the transaction $return$
does not appear on the blockchain, the user $Q$ broadcasts transaction $Fuse$ and get his money back.
\end{itemize}
\section{Security and Performance Analysis}
Because the search process of the second scheme shown in section $4.2$ is similar to that in section $4.1$ and it can run on the current Bitcoin test chain, therefore, in this section, we only give the performance and security analyses for the second scheme.
\begin{figure}
\centering
\begin{minipage}[c]{0.5\textwidth}
\centering
\includegraphics[height=4.5cm,width=7.5cm]{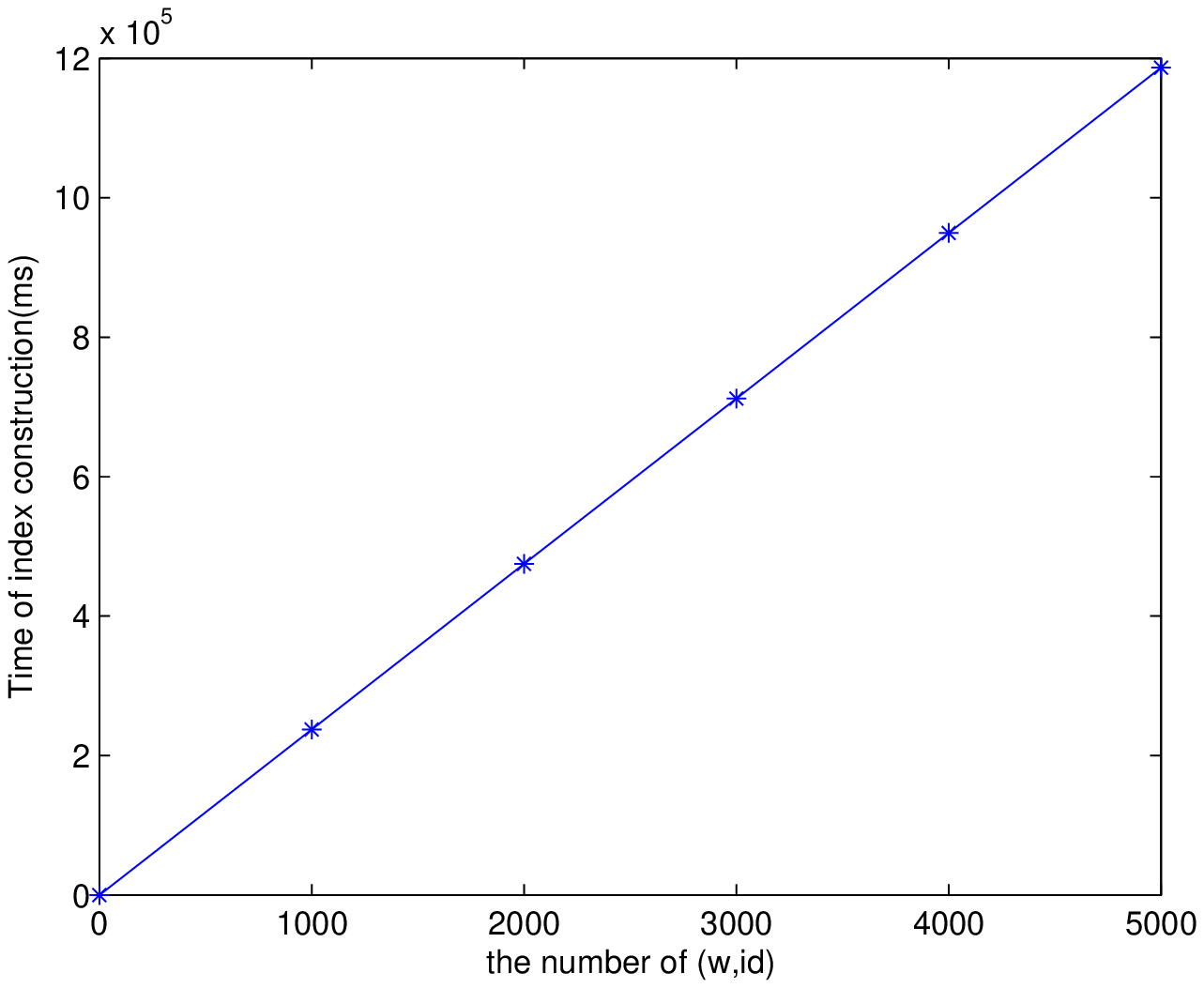}
\end{minipage}%
\begin{minipage}[c]{0.5\textwidth}
\centering
\includegraphics[height=4.5cm,width=7.5cm]{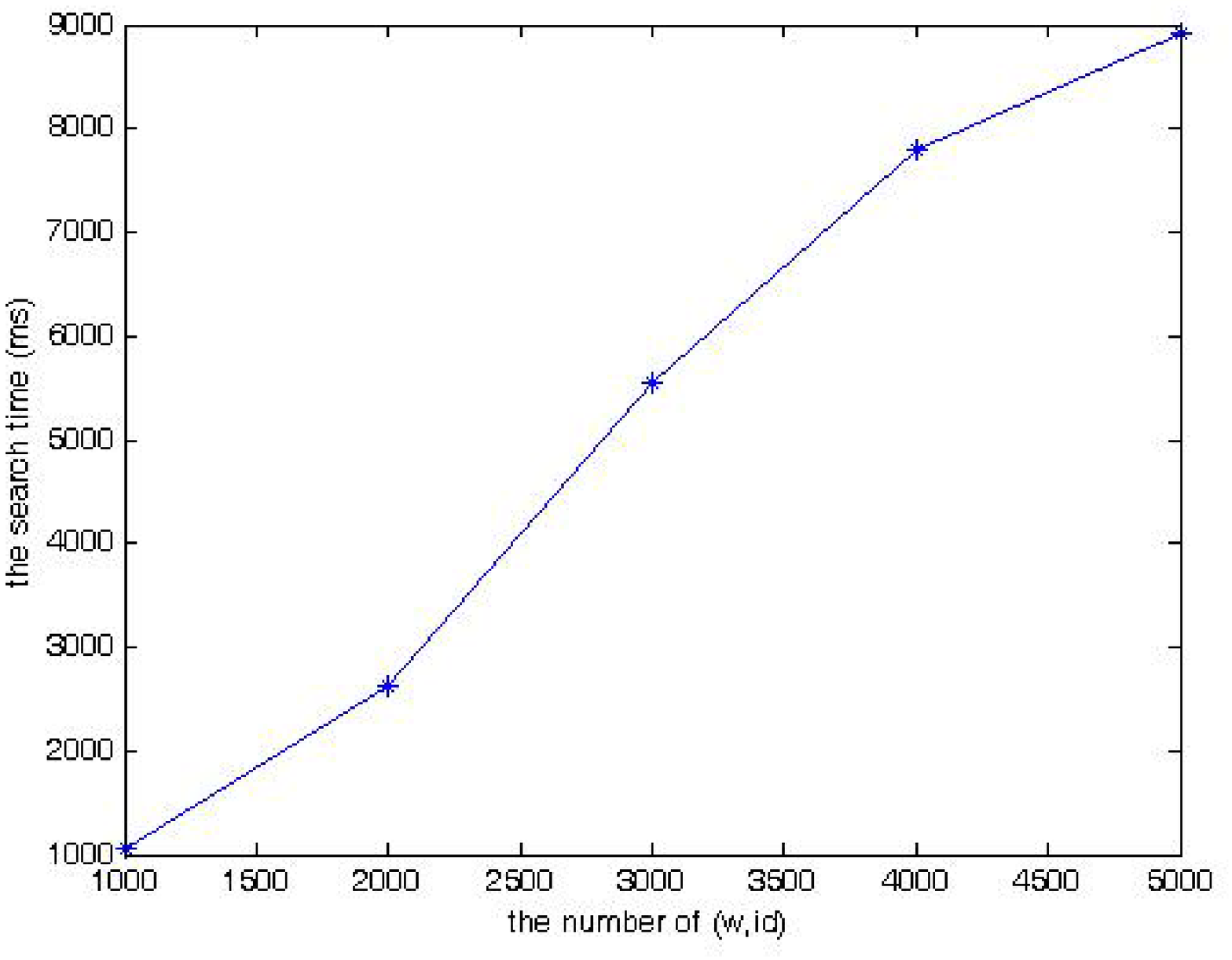}
\end{minipage}
\caption{The time of index construction and search}
\end{figure}

\subsection{Performance}
 We evaluate the performance of our second scheme on the computer with Intel Core $E3-1241 V3$ processor, $32$GB memory. Our code was complied without any optimization in Ubuntu $16.04$ LTS. The version number of Bitcoin that we use is $130200$. In this system, the protocol version number is $70015$, and the version of wallet is $130000$. Each block is produced about $39ms$.

 It is well known that it only supports $80$ bytes in the OP$\_$RETURN script of a transaction, if we directly put $r_i=\epsilon.Enc(K_{11},t_{w_i}\|e_{w_i}\|h_{w_i}\|TI_{w_{i-1}})$ into the out-script, it will result in the transaction is error. Therefore, in order to make the transaction operate normally, we need to divide $(t_{w_i}\|e_{w_i}\|h_{w_i}\|TI_{w_{i-1}})$ into three parts and put them into three transactions without encryption. Namely, for keyword $w_1$, the first part $(h_{w_1}\|0^p)$ will be embedded into transaction $T_{w11}$. After transaction $T_{w11}$ appearing on the blockchain, we put $e_{w_1}$ concatenated with the transaction identifier $TX_{w11}$ of $T_{w11}$ into the second transaction $T_{w12}$. If the transaction $T_{w12}$ appears on the blockchain, we put $t_{w_1}$ concatenated with the transaction identifier $TX_{w12}$ of $T_{w12}$ into the following transaction $T_{w13}$. For keyword $w_i (2\leq i\leq m)$, we put $(h_{w_i}\|TX_{w13})$ into the out-script of transaction $T_{wi1}$, where $TX_{w13}$ denotes the transaction identifier of $T_{w13}$. After transaction $T_{wi1}$ appearing on the blockchain, the $e_{w_i}$ concatenated with the transaction identifier $TX_{wi1}$ of $T_{wi1}$ is put into the second transaction $T_{wi2}$. If the transaction $T_{wi2}$ appears on the blockchain, we put $t_{w_i}$ concatenated with the transaction identifier $TX_{wi2}$ of $T_{wi2}$ into the following transaction $T_{wi3}$. Here, the reason we do not encrypt them is that we want to meet the requirement of Bitcoin test chain, and prove our scheme is feasible. In fact, in order to meet the security requirement, we can modify the parameters to make the transaction support more bytes.

 The parameters related to the time of index construction are the number of documents and keywords. As shown in Fig.8, we can see that the time of index construction is linear into the number of pair $(w,id)$, where $w$ denotes the keyword and $id$ denotes a document's identifier. Because these transactions connect one by one, when search, we need to read from the last transaction until $t_{\tilde{w}}=t_w$ happening. As shown in Fig.9, it shows that the search time is sub-linear in the number of pair $(w,id)$ when we find the documents that only contain keyword $w$ on different scales of data.

\subsection{Security Analysis}
Though the schemes shown in the section 4 are different, their ideas are similar. Therefore, their security proof is similar. In this section, we only give the security proof for the first scheme in the section $4.1$, the proof for the second scheme can be derived similarly.
\begin{theorem}
If {$F_1,F_2,F_3$} are pseudorandom functions, $H$ is a collision resistant hash function, and $\varepsilon=(Enc,Dec)$ is PCPA-secure symmetric encryption scheme, then the scheme we present in section $4.1$ is adaptively IND-CKA2 secure.
\end{theorem}
\begin{proof}
We need to construct a $PPT$ simulator $\mathcal{S}=\{\mathcal{S}_0,\mathcal{S}_1,$ $...,\mathcal{S}_q\}$ and an adversary $\mathcal{A}=\{\mathcal{A}_0,\mathcal{A}_1,$ $...,$ $\mathcal{A}_q\}$ to make the output of $Ideal^\Pi_{{\mathbb{A},\mathbb{S}}}(k)$ and $Real^\Pi_{{\mathbb{A}}}(k)$ be computationally indistinguishable.

Suppose that the simulator $\mathcal{S}$ is given the trace of a history $L=(|T_1|,\ldots,|T_n|,$ $|Inx|,\tau(TX_w))$ where $\tau(TX_w)$ denotes the search pattern and the access pattern about keyword $w$, then he can generate
$(Inx^*,T_1^*,$ $\ldots,T_n^*, Tr, S)$ and claim transaction $ask$ as follows:

\begin{itemize}
\item Simulating $T_1^*,\ldots,T_n^*$.

Because the encryption algorithm $\varepsilon=(Enc,Dec)$ is PCPA-secure, it guarantees the $C_1^*,\ldots,C_n^*$ in $Ideal^\Pi_{{\mathbb{A},\mathbb{S}}}(k)$ game are computationally indistinguishable from the $C_1,\ldots,C_n$ in $Real^\Pi_{{\mathbb{A}}}(k)$ game. Therefore, when the simulator $\mathcal{S}$ embeds $C_1^*,\ldots,C_n^*$ into transactions $T_1^*,\ldots,T_n^*$, they are computationally indistinguishable from the transactions $T_1,\ldots,T_n$ that generated in the $Real^\Pi_{{\mathbb{A}}}(k)  $ game.
    \item Simulating $Inx^*$.

    If $q=0$, $\mathcal S$ sets $t_w^*\leftarrow\{0,1\}^k$, $e_w^*\leftarrow\{0,1\}^k$,$h_w^*\leftarrow\{0,1\}^k$. Therefore, the $t_w,e_w,h_w$ generated in the step $Enc$ in the section 4.1 are computationally indistinguishable from $t_w^*,e_w^*,h_w^*$.

    When $q\geq 1$, $\mathcal S$ selects $l_{w_q}^*\leftarrow\{0,1\}^k$ and $k_{w_q}^*\leftarrow\{0,1\}^k$, then does $ e_{w_q}^*=\epsilon.Enc(l_{w_q}^*,DB^*(w_q))$, $h_{w_q}^*=H(k_{w_q}^*,C_{w_q1}^*\|\ldots \|C_{w_qn}^*)$. Because $F_2,F_3$ are pseudorandom functions, the $(e_{w_q}^*,h_{w_q}^*)$ is computationally indistinguishable from $(e_{w_q},h_{w_q})$ generated in step $Enc$.
    Because $F_1$ is a pseudorandom function, the $t_{w_q}$ generated in step $Enc$ is computationally indistinguishable from $t_{w_q}^*$ that $\mathcal S$ chooses at random from $\{0,1\}^k$.

Therefore, $Inx^*$ is computationally indistinguishable from $Inx$.
    \item Simulating $Tr$. In the transaction $Tr^*$, it embeds $t_w^*$ and $TX_{inx}$. Because $TX_{inx}$ is broadcasted to each other, $\mathcal A$ can get it easily. Here we only consider $t_w^*$ is indistinguishable from $t_{w}$. It uses the pseudorandom function $F_1$ to generate $t_w$ for keyword $w$ in the step $Trpdr$ in the section 4.1, and $t_w$ is indistinguishable from $t_w^*\leftarrow\{0,1\}^k$ that $\mathcal S$ chooses at random. Therefore, $Tr^*$ is computationally indistinguishable from $Tr$.
    \item Claiming the transaction $ask$ by using transaction $S$.

    When $q=0$, if $\mathcal A$ wants to get the money from the transaction $S^*$. $\mathcal S$ returns $(\{C_{i1},\ldots,C_{in}\},h_w)$ to $\mathcal A$, where $C_{ij}\leftarrow\{0,1\}^k (j=1,\ldots,n)$ and $h_w\leftarrow\{0,1\}^k$.
        When $q\geq 1$, $\mathcal S$ firstly returns $(\{C_{wq1},\ldots,C_{wqn}\})$ to $\mathcal A$, where $C_{wqj}(j=1,\ldots,n)$ is the history of access pattern about keyword $w_q$. $\mathcal S$ sets $k_{w_q}^*\leftarrow\{0,1\}^k$ and computes $h_{w_q}^*=H(k_{w_q}^*,C_{wq1}\|\ldots\|C_{wqn})$ which will be sent to $\mathcal A$. Because $F_3$ is a pseudorandom function, therefore the transaction $S$ that the $\mathcal A$ creates cannot claim the money from transaction $ask$.
\end{itemize}
\end{proof}

\section{Conclusion}
At present, the data on the blockchain is increasing, therefore, the search problem becomes more and more serious. Besides, the data on the existing blockchain does not contain anything about the privacy of users, which limits the usage of blockchain. In this paper, we firstly put users' data on the blockchain in the form of encryption. In order to support effective search, we propose two solution with $|O(D(w)|$ complexity, where $|D(w)|$ denotes the number of documents that contain keyword $w$.

Our scheme indirectly solves the search problem on the current blockchains. Moreover, our scheme is very suitable for the medical enterprises, social network and so on. However, the data on these platforms sometimes need to be modified, it means that the index also need be updated. Because the transaction cannot be reversed on the blockchain, which results in the index cannot be updated. Therefore, our scheme is not available for the dynamic data.
\section*{Acknowledgment}
This work is supported by the National Key R\&D Program of China (2017YFB0802503), Natural Science
Foundation of China (61672550), and Natural Science Foundation of Guangdong Province, China (2015A030313133), Fundamental Research Funds for the Central Universities (No.17lgjc45).


\begin{thebibliography}{10}
\providecommand{\url}[1]{\texttt{#1}}
\providecommand{\urlprefix}{URL }

\bibitem{aldermanmulti}
Alderman, J., Martin, K.M., Renwick, S.L.: Multi-level access in searchable
  symmetric encryption (2017), \url{http://eprint.iacr.org/2017/211}

\bibitem{andrychowicz2014fair}
Andrychowicz, M., Dziembowski, S., Malinowski, D., Mazurek, {\L}.: Fair
  two-party computations via bitcoin deposits. In: International Conference on
  Financial Cryptography and Data Security. pp. 105--121. Springer (2014)

\bibitem{andrychowicz2014secure}
Andrychowicz, M., Dziembowski, S., Malinowski, D., Mazurek, L.: Secure
  multiparty computations on bitcoin. In: 2014 IEEE Symposium on Security and
  Privacy. pp. 443--458. IEEE (2014)

\bibitem{bentov2014use}
Bentov, I., Kumaresan, R.: How to use bitcoin to design fair protocols. In:
  International Cryptology Conference. pp. 421--439. Springer (2014)

\bibitem{boldyreva2014efficient}
Boldyreva, A., Chenette, N.: Efficient fuzzy search on encrypted data. In: Fast
  Software Encryption. pp. 613--633. Springer (2014)

\bibitem{bosch2015survey}
B{\"o}sch, C., Hartel, P., Jonker, W., Peter, A.: A survey of provably secure
  searchable encryption. ACM Computing Surveys (CSUR)  47(2), ~18 (2015)

\bibitem{bost2016ovarphiovarsigma}
Bost, R.: ¡Æ o$\varphi$o$\varsigma$: Forward secure searchable encryption. In:
  Proceedings of the 2016 ACM SIGSAC Conference on Computer and Communications
  Security. pp. 1143--1154. ACM (2016)

\bibitem{bost2016verifiable}
Bost, R., Fouque, P., Pointcheval, D.: Verifiable dynamic symmetric searchable
  encryption: Optimality and forward security. {IACR} Cryptology ePrint Archive
   2016, ~62 (2016), \url{http://eprint.iacr.org/2016/062}

\bibitem{buterin2014next}
Buterin, V., et~al.: A next-generation smart contract and decentralized
  application platform. white paper.
  \url{https://github.com/ethereum/wiki/wiki/White-Paper}

\bibitem{camenisch2005compact}
Camenisch, J., Hohenberger, S., Lysyanskaya, A.: Compact e-cash. In: Annual
  International Conference on the Theory and Applications of Cryptographic
  Techniques. pp. 302--321. Springer (2005)

\bibitem{cao2014privacy}
Cao, N., Wang, C., Li, M., Ren, K., Lou, W.: Privacy-preserving multi-keyword
  ranked search over encrypted cloud data. {IEEE} Trans. Parallel Distrib.
  Syst.  25(1),  222--233 (2014)

\bibitem{cash2014dynamic}
Cash, D., Jaeger, J., Jarecki, S., Jutla, C.S., Krawczyk, H., Rosu, M.C.,
  Steiner, M.: Dynamic searchable encryption in very-large databases: Data
  structures and implementation. In: NDSS. vol.~14, pp. 23--26. Citeseer (2014)

\bibitem{cash2013highly}
Cash, D., Jarecki, S., Jutla, C.S., Krawczyk, H., Rosu, M., Steiner, M.:
  Highly-scalable searchable symmetric encryption with support for boolean
  queries. In: Advances in Cryptology--CRYPTO 2013, pp. 353--373. Springer
  (2013)

\bibitem{chaum1983blind}
Chaum, D.: Blind signatures for untraceable payments. In: Advances in
  cryptology. pp. 199--203. Springer (1983)

\bibitem{chaum1984blind}
Chaum, D.: Blind signature system. In: Advances in cryptology. pp. 153--153.
  Springer (1984)

\bibitem{cheng2015verifiable}
Cheng, R., Yan, J., Guan, C., Zhang, F., Ren, K.: Verifiable searchable
  symmetric encryption from indistinguishability obfuscation. In: Proceedings
  of the 10th ACM Symposium on Information, Computer and Communications
  Security. pp. 621--626. ACM (2015)

\bibitem{cui2017secure}
Cui, S., Asghar, M.R., Galbraith, S.D., Russello, G.: Secure and practical
  searchable encryption: A position paper. In: Australasian Conference on
  Information Security and Privacy. pp. 266--281. Springer (2017)

\bibitem{curtmola2006searchable}
Curtmola, R., Garay, J.A., Kamara, S., Ostrovsky, R.: Searchable symmetric
  encryption: improved definitions and efficient constructions. In: Proceedings
  of the 13th ACM conference on Computer and communications security. pp.
  79--88. ACM (2006)

\bibitem{Dai2016Memory}
Dai, S., Li, H., Zhang, F.: Memory leakage-resilient searchable symmetric
  encryption. Future Generation Comp. Syst.  62,  76--84 (2016)

\bibitem{goh2003secure}
Goh, E.: Secure indexes. {IACR} Cryptology ePrint Archive  2003,  216 (2003),
  \url{http://eprint.iacr.org/2003/216}

\bibitem{goldreich1996software}
Goldreich, O., Ostrovsky, R.: Software protection and simulation on oblivious
  rams. Journal of the ACM (JACM)  43(3),  431--473 (1996)

\bibitem{golle2004secure}
Golle, P., Staddon, J., Waters, B.: Secure conjunctive keyword search over
  encrypted data. In: International Conference on Applied Cryptography and
  Network Security. pp. 31--45. Springer (2004)

\bibitem{Ishai2006Cryptography}
Ishai, Y., Kushilevitz, E., Ostrovsky, R., Sahai, A.: Cryptography from
  anonymity. In: Foundations of Computer Science, 2006. FOCS '06. IEEE
  Symposium on. pp. 239--248 (2006)

\bibitem{kamara2017boolean}
Kamara, S., Moataz, T.: Boolean searchable symmetric encryption with worst-case
  sub-linear complexity. In: Annual International Conference on the Theory and
  Applications of Cryptographic Techniques. pp. 94--124. Springer (2017)

\bibitem{kamara2013parallel}
Kamara, S., Papamanthou, C.: Parallel and dynamic searchable symmetric
  encryption. In: International Conference on Financial Cryptography and Data
  Security. pp. 258--274. Springer (2013)

\bibitem{kamara2012dynamic}
Kamara, S., Papamanthou, C., Roeder, T.: Dynamic searchable symmetric
  encryption. In: Proceedings of the 2012 ACM conference on Computer and
  communications security. pp. 965--976. ACM (2012)

\bibitem{kurosawa2012uc}
Kurosawa, K., Ohtaki, Y.: {UC}-secure searchable symmetric encryption. In:
  Financial Cryptography and Data Security, pp. 285--298. Springer (2012)

\bibitem{Kurosawa2013How}
Kurosawa, K., Ohtaki, Y.: How to update documents verifiably in searchable
  symmetric encryption. In: Cryptology and Network Security - 12th
  International Conference, {CANS} 2013, Paraty, Brazil, November 20-22. 2013.
  Proceedings. pp. 309--328 (2013)

\bibitem{lavery1999smart}
Lavery, K.: Smart contracting for local government services: Processes and
  experience. Greenwood Publishing Group (1999)

\bibitem{Li2017Deniable}
Li, H., Zhang, F., Fan, C.: Deniable searchable symmetric encryption.
  Information Sciences  402,  233--243 (2017)

\bibitem{Li2010Fuzzy}
Li, J., Wang, Q., Wang, C., Cao, N.: Fuzzy keyword search over encrypted data
  in cloud computing. In: INFOCOM, 2010 Proceedings IEEE. pp. 1--5 (2010)

\bibitem{moataz2013boolean}
Moataz, T., Shikfa, A.: Boolean symmetric searchable encryption. In:
  Proceedings of the 8th ACM SIGSAC symposium on Information, computer and
  communications security. pp. 265--276. ACM (2013)

\bibitem{nakamoto2008bitcoin}
Nakamoto, S.: Bitcoin: A peer-to-peer electronic cash system (2008),
  \url{http://www.cryptovest.co.uk/resources/Bitcoin 20paper 20Original.pdf}

\bibitem{naveed2014dynamic}
Naveed, M., Prabhakaran, M., Gunter, C.A.: Dynamic searchable encryption via
  blind storage. In: Security and Privacy (SP), 2014 IEEE Symposium on. pp.
  639--654. IEEE (2014)

\bibitem{ron2013quantitative}
Ron, D., Shamir, A.: Quantitative analysis of the full bitcoin transaction
  graph. In: International Conference on Financial Cryptography and Data
  Security. pp. 6--24. Springer (2013)

\bibitem{song2000practical}
Song, D.X., Wagner, D., Perrig, A.: Practical techniques for searches on
  encrypted data. In: Security and Privacy, 2000. S\&P 2000. Proceedings. 2000
  IEEE Symposium on. pp. 44--55. IEEE (2000)

\bibitem{stefanov2014practical}
Stefanov, E., Papamanthou, C., Shi, E.: Practical dynamic searchable encryption
  with small leakage. In: NDSS. vol.~14, pp. 23--26 (2014)

\bibitem{swan2015blockchain}
Swan, M.: Blockchain: Blueprint for a new economy. " O'Reilly Media, Inc."
  (2015), \url{http://www.oreilly.com/catalog/errata.csp?isbn=9781491920497}.

\bibitem{van2010computationally}
Van~Liesdonk, P., Sedghi, S., Doumen, J., Hartel, P., Jonker, W.:
  Computationally efficient searchable symmetric encryption. In: Workshop on
  Secure Data Management. pp. 87--100. Springer (2010)

\bibitem{wong2009secure}
Wong, W.K., Cheung, D.W.l., Kao, B., Mamoulis, N.: Secure knn computation on
  encrypted databases. In: Proceedings of the 2009 ACM SIGMOD International
  Conference on Management of data. pp. 139--152. ACM (2009)
\bibitem{FuHRWW17}
Z.~Fu, F.~Huang, K.~Ren, J.~Weng, C.~Wang, Privacy-preserving smart semantic
  search based on conceptual graphs over encrypted outsourced data, {IEEE}
  Trans. Information Forensics and Security 12~(8) (2017) 1874--1884.

\bibitem{guo2017dynamic}
C.~Guo, C.~Xue, Y.~ Jie, Zhang.~Fu, M.~Li, B.~Feng, Dynamic Multi-phrase Ranked Search over Encrypted Data with Symmetric Searchable Encryption, {IEEE}
  Trans. Services Computing (2017),\url{http://ieeexplore.ieee.org/abstract/document/8089767/}.
\end{thebibliography}
\end{document}